\RequirePackage{fix-cm}
\documentclass[twocolumn]{svjour3}          
\smartqed  
\usepackage{graphicx}
\usepackage{verbatim}
\usepackage[authoryear]{natbib}
\usepackage{amsmath,amssymb,amsfonts}
\journalname{Statistics and Computing}
\newcommand{\field}[1]{\mathbb{#1}}
\newcommand{\Z}{\field{Z}}

\DeclareMathOperator{\xRe}{Re}
\DeclareMathOperator{\xlcm}{lcm}
\newcommand{\design}{{\mathcal D}}
\newcommand{\fraction}{{\mathcal F}}

\begin{document}

\title{Generalized Minimum Aberration mixed-level orthogonal arrays
}
\subtitle{A general approach based on sequential integer quadratically constrained quadratic programming}



\author{Roberto Fontana}



\institute{R. Fontana \at
              Department of Mathematical Sciences - Politecnico di Torino \\
              Tel.: +39-011-0907504
              \email{roberto.fontana@polito.it}           
}

\date{Received: date / Accepted: date}

\maketitle

\begin{abstract}
Orthogonal Fractional Factorial Designs and in particular Orthogonal Arrays are frequently used in many fields of application, including medicine, engineering and agriculture. 
In this paper we present a methodology and an algorithm to find an orthogonal array, of given size and strength, that satisfies the generalized minimum aberration criterion.  The methodology is based on the joint use of polynomial counting functions, complex coding of levels and algorithms for quadratic optimization and puts no restriction on the number of levels of each factor. 

\keywords{Design of experiments \and Generalized minimum aberration criterion \and Generalized wordlength pattern \and Counting function \and integer quadratically constrained quadratic programming}
\end{abstract}

\section{Introduction}
\label{intro}
In this paper we present a methodology to find one of the \emph{best} orthogonal arrays for the generalized minimum aberration (GMA) criterion, as defined in \cite{cheng2004geometric}. We refer to these designs as GMA-optimal designs. For an $m$-factor design, the GMA-criterion is to sequentially minimize the severity of aliasing between all the $i$-factor effects and the overall mean, starting from $i=1$ (main effects) and finishing at $i=m$ ($m$-factor interaction effects).

The joint use of polynomial indicator functions and complex coding of levels provides a general theory for mixed level orthogonal fractional factorial designs, see \cite{pistone|rogantin:2008-JSPI-3}. It also makes use of commutative algebra, see \cite{pistone|wynn:96}, and generalizes the approach to two-level designs as discussed in \cite{fontana|pistone|rogantin:2000}. 
This theory does not put any restriction either on the number of levels of each factor or on the orthogonality constraints. 
It follows that our methodology can be applied to find any GMA-optimal mixed-level orthogonal array.

Orthogonal Arrays (OAs) are frequently used in many fields of application, including medicine, engineering and agriculture. They offer a valuable tool for dealing with problems where there are many factors involved and each run is expensive. They also keep  the statistical analysis of the data quite simple. The literature on the subject is extremely rich. A non-exhaustive list of references, mainly related to the theory of the design of experiments, includes the fundamental paper of \cite{bose:47} and the following books: \cite{raktoe|hedayat|federer:81}, \cite{collombier:96}, \cite{dey|mukerjee:99}, \cite{wu|hamada:2000}, \cite{mukerjee|wu:2006} and  \cite{bailey:08}.

Orthogonal Arrays represent an important class of Orthogonal Fractional Factorial Designs (OFFDs), see, for example, \cite{hedayat|sloane|stufken:1999} and \cite{schoen|eendebak|guyen:2010}. Indeed an Orthogonal Array of appropriate strength can be used to solve the wide range of problems related to the quantification of both the size of the main effects and the interactions up to a given order of interest. 

This paper is organized as follows: in Section~\ref{sec:alg} we briefly review the algebraic theory of OFFDs based on polynomial counting functions. The computation of the wordlength pattern of a design is described in Section~\ref{sec:ab_crit} while we describe the algorithm in Section~\ref{sec:algo}. Some applications of the algorithm are presented in Section~\ref{sec:testcases}. Finally, concluding remarks are in Section~\ref{sec:con}. Section~\ref{sec:alg} is closely based on Section 2 of \cite{fontanacost}. We include it here to facilitate the reader.

\section{Algebraic characterization of OFFDs} \label{sec:alg}
In this Section, for ease in reference, we present some relevant results of the algebraic theory of OFFDs.  The interested reader can find further information, including the proofs of the propositions, in \cite{fontana|pistone|rogantin:2000} and \cite{pistone|rogantin:2008-JSPI-3}.

\subsection{Fractions of a full factorial design} \label{sec:fr_ff}
Let us consider an experiment which includes $m$ factors $\design_{j}, \; j=1,\ldots,m$. Let us code the $n_j$ levels of the factor $\design_{j}$ by  the $n_j$-th roots of the unity
\[
\design_{j} = \{\omega_0^{(n_j)},\ldots,\omega_{n_j-1}^{(n_j)}\},
\]
where $\omega_k^{(n_j)}=\exp\left(\sqrt{-1}\:  \frac {2\pi}{n_j} \ k\right)$, $k=0,\ldots,n_j-1, \ j=1,\ldots,m$.

The \emph{full factorial design} with \emph{complex coding} is $\design = \design_1 \times \cdots \design_j \cdots \times \design_m$. We denote its cardinality by $\# \design$, $\# \design=\prod_{j=1}^m n_j$.

\begin{definition}
A fraction $\fraction$ is a multiset $(\fraction_*,f_*)$ whose underlying set of elements $\fraction_*$ is contained in $\design$ and $f_*$ is the multiplicity function $f_*: \fraction_* \rightarrow \mathbb N$ that for each element in $\fraction_*$ gives the number of times it belongs to the multiset $\fraction$. 
\end{definition}
We recall that the underlying set of elements $\fraction_*$ is the subset of $\design$ that contains all the elements of $\design$ that appear in $\fraction$ at least once. We denote the number of elements of a fraction $\fraction$ by $\# \fraction$, with  $\# \fraction= \sum_{\zeta \in \fraction_*} f_*(\zeta)$.

\begin{example}
Let us consider $m=1$, $n_1=3$. We get 
\[\design=\{ 1,\exp\left(\sqrt{-1}\:  \frac{2\pi}{3} \right), \exp\left(\sqrt{-1}\:  \frac{4\pi}{3} \right)\}.
\]
The fraction $\fraction=\{1,1,\exp\left(\sqrt{-1}\:  \frac {2\pi}{3} \right)\}$ is the multiset $(\fraction_*,f_*)$ where 
$\fraction_*=\{1,\exp\left(\sqrt{-1}\:  \frac {2\pi}{3} \right)\}$, $f_*(1)=2$, and 
$f_* (\exp\left(\sqrt{-1}\:  \frac{2\pi}{3} \right)) =1$. We get $\# \fraction=f_*(1)+f_* (\exp\left(\sqrt{-1}\:  \frac{2\pi}{3} \right))=2+1=3$.
\end{example}

In order to use polynomials to represent all the functions defined over $\design$, including multiplicity functions, we define
\begin{itemize} 
\item $X_j$, the $j$-th component function, which maps a point $\zeta=(\zeta_1,\ldots,\zeta_m)$ of $\design$ to its $j$-th component,
\[
X_j \colon \design \ni (\zeta_1,\ldots,\zeta_m)\ \longmapsto \ \zeta_j \in \design_j \ .
\]
The function $X_j$ is called \emph{simple term} or, by abuse of terminology, \emph{factor}.
\item $X^\alpha=X_1^{\alpha_1} \cdot \ldots \cdot X_m^{\alpha_m}$, $\alpha \in L =  \Z_{n_1} \times \cdots \times  \Z_{n_m}$ i.e., the monomial function
\[
X^\alpha : \design \ni (\zeta_1,\ldots,\zeta_m)\ \mapsto \ \zeta_1^{\alpha_1}\cdot \ldots \cdot \zeta_m^{\alpha_m} \ .
\] 
The function $X^\alpha$ is called \emph{interaction term}
\end{itemize}

We observe that $\{X^\alpha: \alpha \in L =  \Z_{n_1} \times \cdots \times  \Z_{n_m}\}$ is a basis of all the complex functions defined over $\design$. We use this basis to represent the counting function of a fraction according to Definition \ref{de:indicator}.

\begin{definition} \label{de:indicator}

The \emph{counting function} $R$ of a fraction $\fraction$ is a complex polynomial defined over $\design$ so that for each $\zeta \in \design$, $R(\zeta)$ equals the number of appearances of $\zeta$ in the fraction. A $0-1$ valued counting function is called an \emph{indicator function} of a single replicate fraction $\fraction$.
We denote by $c_\alpha$ the coefficients of the representation of $R$  on $\design$ using the monomial basis
$\{X^\alpha, \ \alpha \in L\}$:
$$
R(\zeta) = \sum_{\alpha \in L} c_\alpha X^\alpha(\zeta), \;\zeta\in\design, \;  c_\alpha \in \mathbb C \ .
$$
\end{definition}

%

With Proposition \ref{pr:bc-alpha} from \cite{pistone|rogantin:2008-JSPI-3}, we link the orthogonality of two interaction terms with the coefficients of the polynomial representation of the counting function. We denote by $\overline{x}$ the complex conjugate of the complex number $x$.
   
\begin{proposition} \label{pr:bc-alpha}
If $\fraction$ is a fraction of a full factorial design $\design$, $R = \sum_{\alpha \in L} c_\alpha X^\alpha$ is its counting function and $[\alpha-\beta]$ is the $m$-tuple made by the componentwise difference in the rings $\Z_{n_j}$,  
 $\left(\left[\alpha_1-\beta_1 \right]_{n_1}, \ldots, \left[\alpha_m - \beta_m\right]_{n_m} \right)$, then
\begin{enumerate}
 \item \label{it:balpha}
the coefficients $c_\alpha$ are given by $c_\alpha= \frac 1 {\#\design} \sum_{\zeta \in \fraction} \overline{X^\alpha(\zeta)} \ ;$
 \item \label{it:cent1}
 the term $X^\alpha$ is centered on $\fraction$ i.e., $\frac{1}{\#\fraction} \sum_{\zeta \in \fraction} X^\alpha(\zeta)=0$ if, and only if,
 $c_\alpha=c_{[-\alpha]}=0 $;
 \item \label{it:cent2}
 the terms $X^\alpha$ and $X^\beta$
are orthogonal on $\fraction$ if and only if,
$c_{[\alpha-\beta]}=0 $.
\end{enumerate}
\end{proposition}

We now define projectivity and, in particular, its relation with orthogonal arrays. Given $I=\{i_1,\ldots,i_k\} \subset \{1,\ldots,m\}, i_1<\ldots < i_k$ and $\zeta=(\zeta_1,\ldots,\zeta_m) \in \design$ we define the projection $\pi_I(\zeta)$ as
\[
\pi_I(\zeta)=\zeta_I \equiv (\zeta_{i_1},\ldots,\zeta_{i_k}) \in \design_{i_1} \times \ldots \times \design_{i_k} \ .
\]

\begin{definition}
A fraction $\fraction$ {\em factorially projects} onto the $I$-factors, $I=\{i_1,\ldots,i_k\} \subset \{1,\ldots,m\}$, $i_1<\ldots < i_k$, if the projection $\pi_I(\fraction)$ is 
a multiple full factorial design, i.e., the multiset $(\design_{i_1} \times \ldots \times \design_{i_k} , f_*)$ where the multiplicity function $f_*$ is constantly equal to a positive integer over $\design_{i_1} \times \ldots \times \design_{i_k}$.
\end{definition}

\begin{example}
Let us consider $m=2, n_1=n_2=2$ and the fraction
$\fraction=\{ (-1,1), (-1,1),$ 
\newline
$(1,-1), (1,1)\}$.
We obtain $\pi_1(\fraction)=\{-1,-1,1,1\}$ and $\pi_2(\fraction)=\{-1,1,1,1\}$. It follows that $\fraction$ projects on the first factor and does not project on the second factor. 
\end{example}

\begin{definition}
A fraction $\fraction$ is a {\em mixed orthogonal
array} of strength $t$ if it factorially projects onto any $I$-factors with $\#I=t$.
\end{definition}

  
\begin{proposition} \label{pr:projectivity}
A fraction \emph{factorially projects onto the $I$-factors}, \\ 
$I=\{i_1,\ldots,i_k\} \subset \{1,\ldots,m\}, i_1<\ldots < i_k$, if and only if,
all the coefficients of the counting function involving the $I$-factors only are $0$.
\end{proposition}

Proposition \ref{pr:projectivity} can be immediately stated for mixed orthogonal arrays.
\begin{proposition} \label{pr:projectivity_ort}
A fraction is an \emph{orthogonal array of strength $t$}  if and
only if, all the coefficients $c_{\alpha}, \; \alpha \neq (0,\ldots,0)$ of the counting function up to the order $t$ are $0$.
\end{proposition}

\section{Aberration criterion}
\label{sec:ab_crit}

Using the polynomial counting function, \cite{cheng2004geometric} provide the following definition of the generalized wordlength pattern $\alpha_\fraction=(\alpha_1(\fraction), \ldots, \alpha_m(\fraction))$ of a fraction $\fraction$ of the full factorial design $\design$.

\begin{definition} \label{wlp}
The generalized wordlength pattern $\alpha_\fraction=(\alpha_1(\fraction), \ldots, \alpha_m(\fraction))$ of a fraction $\fraction$ of the full factorial design $\design$ is defined as
\[
\alpha_i(\fraction)= \sum_{\|\alpha \|_0 =i} \left( \frac{c_{\alpha}}{c_{0}}  \right)^2 \quad i=1,\ldots,m
\]
\end{definition}
where $\| \alpha \|_0$ is the number of non-null elements of $\alpha$.

According to the algebraic methodology that we have described in Section \ref{sec:alg}, as $c_\alpha$'s are complex numbers, we should simply generalize Definition \ref{wlp} as follows.

\begin{definition} \label{wlp1}
The generalized wordlength pattern $\alpha_\fraction=(\alpha_1(\fraction),\ldots,\alpha_m(\fraction))$ of a fraction $\fraction$ of the full factorial design $\design$ is defined as
\[
\alpha_i(\fraction)= \sum_{\|\alpha \|_0 =i} \left(  \frac{ \|c_{\alpha}\|_2 }{\|c_{0}\|_2} \right)^2  \quad i=1,\ldots,m
\]
\end{definition}
where $\| x \|_2$ is the norm of the complex number $x$.

The generalized minimum aberration criterion is to sequentially minimize $\alpha_i(\fraction)$ for $i=1,\ldots,m$.

In Section \ref{norm} we provide a formula to compute $\alpha_i(\fraction), \; i=1,\ldots,m$, given a fraction $\fraction \subseteq \design$.

\subsection{The wordlength patterm of a fraction} \label{norm}
Given a $\fraction$ of the full factorial design $\design$, let us consider its counting function $R = \sum_{\alpha \in L} c_\alpha X^\alpha$. From item \ref{it:balpha} of Proposition \ref{pr:bc-alpha} the coefficients $c_\alpha$ are given by 
\[
c_\alpha= \frac 1 {\#\design} \sum_{\zeta \in \fraction} \overline{X^\alpha(\zeta)} 
\]
or equivalently
\[
c_\alpha= \frac 1 {\#\design} \sum_{\zeta \in \design} R(\zeta) \overline{X^\alpha(\zeta)} .
\]

The square of the norm of a complex number $x$ can be computed as $ x \overline {x}$. It follows that
\[
\| c_\alpha \|_2^2 =  c_\alpha \overline{ c_\alpha}
\]
To make the notation easier we make the non-restrictive hypothesis that both the runs $\zeta$ of the full factorial design $\design$ and the multi-indexes of $L=\Z_{n_1} \times \cdots \times  \Z_{n_m}$ are considered in lexicographic order. We get
\begin{eqnarray*}
(\#\design) \| c_\alpha \|_2^2 = \sum_{\zeta \in \design} R(\zeta) \overline{X^\alpha(\zeta)} =\\
 = ( \overline {X^\alpha}^T Y )\overline{( \overline {X^\alpha}^T Y )} = Y^T \overline {X^\alpha} {X^\alpha}^{T} Y
\end{eqnarray*}
where $X^\alpha$ is the column vector $\left[ \zeta^\alpha : \zeta \in \design \right]$, $Y$ is the column vector $\left[ R(\zeta) : \zeta \in \design \right]$ and $.^T$ denote the transpose of a vector. 
We refer to $Y$ as the counting vector of a fraction and we denote by $H^{\alpha}=[ h_{ij}: i,j=1,\ldots,\#\design ]$ the matrix $\overline {X^\alpha} {X^\alpha}^{T}$. By construction the matrix $H^{\alpha}$ is Hermitian and positive-definite.

\begin{proposition}
The square of the norm of $c_\alpha$ is
\[
\|c_\alpha\|_2^2= \frac{1}{(\#\design)^2} Y^T  H_R^{\alpha} Y
\]
where $H_R^{\alpha}=[ \xRe(h_{i,j}): i,j=1,\ldots,\#\design]$ and $\xRe(h_{i,j})$ it the real part of the complex number $h_{i,j}$.
\end{proposition}
\begin{proof}
For a quadratic form we have
\[
Y^T  H^{\alpha} Y = Y^T  {(H^{\alpha})}^T Y 
\]
The matrix $H^{\alpha}$ is Hermitian: ${(H^{\alpha})}^T = \overline{H^{\alpha}}$.
It follows that 
\[
(\#\design) \| c_\alpha \|_2^2 = Y^T H_R^{\alpha} Y 
\] 
where $H_R^{\alpha}=[ \xRe(h_{i,j}): i,j=1,\ldots,\#\design]$ and $\xRe(h_{i,j})$ it the real part of the complex number $h_{i,j}$.
\qed
\end{proof}

In this way we can compute the generalized word length pattern using only real valued vectors and matrices. In Proposition \ref{pr:hr} we provide an explicit expression of the elements of the matrix $H_R^{\alpha}$.

\begin{proposition} \label{pr:hr}
The real part of the element $h_{i,j}$ of the matrix $H^\alpha$ is
\[
\cos \left( \frac {2\pi}{n} \sum_{k=1}^{m} \frac{n}{n_k} \alpha_k (t_k-z_k) \right) \quad i,j=1,\ldots,\#\design
\]
where $(z_1,\ldots,z_m)$ (resp. $(t_1,\ldots,t_m)$ )  is the $i$-th (resp. $j$-th) row of $L=\Z_{n_1} \times \cdots \times  \Z_{n_m}$ and $n$ is the lowest common multiple of $n_1,\ldots,n_m$, $n=\xlcm(n_1,\ldots,n_m)$.
\end{proposition}
\begin{proof}
Let $\zeta=(\zeta_1,\ldots,\zeta_m)$ be the $i$-th row of $\design$. We have $\zeta_k= \exp (\sqrt{-1} \frac {2\pi}{n_k} z_k), k=1,\ldots,m$ where $(z_1,\ldots,z_m)$ is the $i$-th row of $L=\Z_{n_1} \times \cdots \times  \Z_{n_m}$. Analogously let $\tau=(\tau_1,\ldots,\tau_m)$ be the $j$-th row of $\design$. We have $\tau_k= \exp (\sqrt{-1} \frac {2\pi}{n_k} t_k), k=1,\ldots,m$ where $(t_1,\ldots,t_m)$ is the $j$-row of $L$. 

The complex conjugate of $\zeta_k$ is $ \overline{\zeta}_k=\exp (- \sqrt{-1} \frac {2\pi}{n_k} z_k)$ $k=1,\ldots,m$.
It follows that $h_{i,j}$ can be written as
\[
\exp ( \sqrt{-1} \frac {2\pi}{n_1} \alpha_1 (t_1-z_1) ) \cdot \ldots \cdot \exp ( \sqrt{-1} \frac {2\pi}{n_m} \alpha_m (t_m-z_m) )
\]
or 
\[
\exp ( \sqrt{-1} \frac {2\pi}{n} ( \frac{n}{n_1} \alpha_1 (t_1-z_1)  + \ldots + \frac {n}{n_m} \alpha_m (t_m-z_m) ) )
\]
where $n$ is the lowest common multiple of $n_1,\ldots,n_m$, $n=\xlcm(n_1,\ldots,n_m)$.
Taking the real part of $h_{i,j}$ we complete the proof.
\qed
\end{proof}
\begin{proposition} \label{pr:wlp}
The generalized wordlength pattern $\alpha_\fraction=(\alpha_1(\fraction),\ldots,\alpha_m(\fraction))$ of a fraction $\fraction$ of the full factorial design $\design$ is
\[
\alpha_i(\fraction)=  \frac{1}{(\#\fraction)^2} Y^T H_i Y \; i=1,\ldots,m
\]
where $H_i = \sum_{\|\alpha\|_{0}=i} H_R^{\alpha}$.
\end{proposition}
\begin{proof}
From Definition \ref{wlp1} we have 
\[
\alpha_i(\fraction) = \sum_{\|\alpha \|_0 =i} \left(  \frac{ \|c_{\alpha}\|_2 }{\|c_{0}\|_2} \right)^2 \quad i=1,\ldots,m.
\]
From item \ref{it:balpha} of Proposition \ref{pr:bc-alpha} we get $c_0=\frac{\#\fraction}{\#\design}$ and therefore $\|c_0 \|_2^2=\left(\frac{\#\fraction}{\#\design}\right)^2$. We can also write 
\begin{eqnarray*}
(\#\design)^2\sum_{\|\alpha \|_0 =i} \|c_{\alpha}\|_2^2= (\#\design)^2\sum_{\|\alpha \|_0 =i} Y^T H_R^\alpha Y= \\
= (\#\design)^2  Y^T (\sum_{\|\alpha \|_0 =i}  H_R^\alpha) Y =  (\#\design)^2  Y^T H_i Y
\end{eqnarray*}
where $H_i = \sum_{\|\alpha\|_{0}=i} H_R^{\alpha}$.
\qed
\end{proof}

From a computational point of view (see Section \ref{GMA}) it is useful to consider the Cholesky decomposition of the symmetric and positive definite matrix $H_i$ 
\[
H_i=U_i^T U_i \quad i=1,\ldots,m.
\] 

Thus from Proposition \ref{pr:wlp} the wordlength pattern of a fraction $\fraction \subseteq \design$ can be written as 
\[
\alpha_i(\fraction) = \frac{1}{(\#\fraction)^2} \|U_i Y \|_2^2 \quad i=1,\ldots,m.
\]   

\subsection{GMA for mixed level orthogonal arrays} \label{GMA}
From Proposition \ref{pr:projectivity_ort} we know that for orthogonal arrays of strength $t$ all the coefficients $c_\alpha$ of the counting function up to order $t$ (that is $0 < \| \alpha \|_0 \leq t$) must be $0$.
The generalized wordlength pattern $\alpha_\fraction$ of an orthogonal array $\fraction$ of strength $t$ will be
\[
( 0, \ldots, 0, \alpha_{t+1}(\fraction), \ldots, \alpha_m(\fraction)) .
\] 
It follows that the counting vectors must satisfy the following condition 
\[
U_i Y = 0 \; \quad i=1,\ldots,t 
\]
or, equivalently,
\[
A_t Y = 0
\]
where
\begin{equation*}
A_t= \left[ 
\begin{array}[h]{c}
	U_1 \\
	\ldots \\
	U_t
\end{array}
\right]
\end{equation*}

\cite{fontanacost}, \cite{fontanapistoneadv13} and \cite{fontanasampo} show that, given the full factorial design $\design$, the counting vectors $Y=[R(\zeta): \; \zeta \in \design]$ of the orthogonal arrays $\fraction \subseteq \design$ of strength $t$ are the positive integer solutions of a system of linear equations, $AY=0$. The present paper shows a different way to build the constraint matrix $A$. 

Let us suppose that we are interested in orthogonal arrays of size $N$. The GMA-criterion will provide the OAs with the maximum strength $t$ for the given size $N$. If the norm of $U_1Y$ is strictly positive, $\|U_1 Y\|_2>0$ for all the counting vectors $Y$ it follows that no OA exists for that size $N$.

\section{An algorithm for GMA designs} \label{sec:algo}
A counting vector of a GMA-optimal design can be obtained through the $m$ steps of the algorithm below. 

\subsection{Input}
The input of the algorithm is made of:
\begin{enumerate}
\item the number $m$ of the factors and the number of level $n_i$ of the $i$-th factor, $i=1,\ldots,m$;
\item the size $N$ of the fraction $\fraction$.
\end{enumerate} 
\subsection{Step $1$}
Solve the following quadratic optimization problem
\begin{equation} \label{eq:opt1}
\begin{cases}
\min \|U_{1} Y \|_2^2 \\
\text{subject to } \\
1^T Y= N, \\
Y=[y_i], y_i \in \Z, y_i \geq 0 
\end{cases}
\ .
\end{equation} 
Let us denote by $Y_{1}^{\star}$ one solution and let $W_{1}^{\star}= \| U_{1} Y_{1}^{\star}\|_2^2$.

\subsection{Step $k$ with $k=2,\ldots,m$}
Solve the following quadratic optimization problem
\begin{equation} \label{eq:optk}
\begin{cases}
\min \|U_{k} Y \|_2^2 \\
\text{subject to } \\
1^T Y= N, \\
\| U_{1} Y \|_2^2 = W_{1}^{\star}, \\
\ldots \\
\| U_{k-1} Y \|_2^2 =W_{k-1}^{\star}, \\ 
Y=[y_i], y_i \in \Z, y_i \geq 0 
\end{cases}
\ .
\end{equation} 
Let us denote by $Y_{k}^{\star}$ one solution and let $W_{k}^{\star}=\| U_{k} Y_{k}^{\star}\|_2^2$. We observe that if $W_{j}^{\star}=0, j \in \{1,\ldots,m\}$ then the condition $W_{j}^{\star}=\| U_{j} Y\|_2^2$ can be simply replaced by $U_{j} Y = 0 $.

\subsection{Output}
If $ \|U_1 Y_{m}^{\star}\|_2=\ldots=\|U_t Y_{m}^{\star}\|_2=0$ and  $\|U_{t+1} Y_{m}^{\star}\|_2>0$ the solution $Y_{m}^{\star}$ of the last optimization problem, corresponding to the $m$-th step, is the counting vector of an orthogonal array $\fraction \subseteq \design$ of size $N$ and strength $t$ that is optimal according to the GMA-criterion. The wordlength pattern of $\fraction$ is
\[
(0,\ldots,0, \frac{1}{(\#\fraction)^2}\| W_{t+1}^{\star} \|_2^2, \ldots, \frac{1}{(\#\fraction)^2}\| W_{m}^{\star} \|_2^2 ).
\]
If $ \|U_1 Y_{m}^{\star}\|_2 >0$ then the solution $Y_{m}^{\star}$ is a fraction that that is optimal according to the GMA-criterion but that is not an OA.

\section{Test cases} \label{sec:testcases}
We denote by $OA(N,n_1 \cdot \ldots \cdot n_m, t)$ a mixed level orthogonal array with $N$ rows, $m$ columns (the $i$-th column has $n_i$ levels, $i=1,\ldots,m$) and with strength $t$. 

The computations are made using 
\begin{itemize}
\item one main module, written in SAS IML, that prepares the $m$ optimization problems, \cite{man1};
\item MOSEK that solves each optimization problem, \cite{rmosek}. 
\end{itemize}
The simulation study has been conducted on a standard laptop (CPU Intel Core i7-2620M CPU 2.70 GHz 2.70 GHz, RAM 8 Gb). 

\subsection{Five factors with $2$ levels each}
We consider five factors with $2$ levels each. We initially choose $N=16$. Using the algorithm described in Section~\ref{sec:alg}, we obtain $5$ fractions, that are the optimal solutions corresponding to steps $1,\ldots,5$. The wordlength patterns of these fractions are shown in Table~\ref{tab:wlp16}. The GMA-optimal fraction, that is found at the last iteration, is an $OA(16,2^5,4)$ with a wordlength pattern equal to $(0,0,0,0,1)$.

\begin{table}
\caption{Wordlength patterns of optimal fractions at different steps; $m=5, n_1=\ldots=n_5=2, N=16$}
\label{tab:wlp16}       
\begin{tabular}{c|ccccc}
\hline\noalign{\smallskip}
step & $\alpha_1(\fraction)$ & $\alpha_2(\fraction)$ & $\alpha_3(\fraction)$ & $\alpha_4(\fraction)$ & $\alpha_5(\fraction)$  \\
\noalign{\smallskip}\hline\noalign{\smallskip}
1 & 0 & 10 & 0 & 5 & 0 \\
2 & 0 & 0 & 1 & 0 & 0 \\
3 & 0 & 0 & 0 & 1 & 0 \\
4 & 0 & 0 & 0 & 0 & 1 \\
5 & 0 & 0 & 0 & 0 & 1 \\
\noalign{\smallskip}\hline
\end{tabular}
\end{table}

We now consider $N=6,8,10,12$. The wordlength pattern of the GMA-optimal solutions are presented in Table~\ref{tab:wlp8_12}. 

\begin{table}
\caption{Wordlength patterns of GMA-optimal fractions of different sizes; $m=5, n_1=\ldots=n_5=2, N=16$}
\label{tab:wlp8_12}       
\begin{tabular}{c|c@{\hskip 0.1in}c@{\hskip 0.1in}c@{\hskip 0.1in}c@{\hskip 0.1in}c@{\hskip 0.1in}c}
\hline\noalign{\smallskip}
N & $\alpha_1(\fraction)$ & $\alpha_2(\fraction)$ & $\alpha_3(\fraction)$ & $\alpha_4(\fraction)$ & $\alpha_5(\fraction)$ & Type \\
\noalign{\smallskip}\hline\noalign{\smallskip}
6 & 0 & 1.11 & 1.78 & 1.44 & 0 & $OA(6,2^5,1)$ \\
8 & 0 & 0 & 2 & 1 & 0 & $OA(8,2^5,2)$ \\
10 & 0  &   0.4   &   0  &    1.8   &   0 & $OA(10,2^5,1)$ \\
12 & 0 & 0 & 1.11 & 0.56 & 0 & $OA(12,2^5,2)$ \\
14 & 0 & 0.2 & 0 & 1.08 & 0 & $OA(14,2^5,1)$ \\
\noalign{\smallskip}\hline
\end{tabular}
\end{table}

\subsection{Four factors, one with $2$ levels and three with $3$ levels each}
We consider four factors, the first with $2$ levels and the remaining three with $3$ levels each.
We initially choose $N=18$. We obtain $4$ fractions, that are the optimal solutions obtained at step $1,\ldots,4$. The wordlength patterns of these fractions are reported in Table~\ref{tab:wlp18}. The GMA-optimal fraction, that is found at the last iteration, is an $OA(18,2\cdot3^3,2)$ with a wordlength pattern equal to $(0,0,0.5,1.5)$. 

\begin{table}
\caption{Wordlength patterns of optimal fractions; $m=4, n_1=2, n_2=\ldots=n_4=3, N=18$}
\label{tab:wlp18}       
\begin{tabular}{c|cccc}
\hline\noalign{\smallskip}
step & $\alpha_1(\fraction)$ & $\alpha_2(\fraction)$ & $\alpha_3(\fraction)$ & $\alpha_4(\fraction)$  \\
\noalign{\smallskip}\hline\noalign{\smallskip}
1 & 0 & 1.81 & 1.09 & 1.09 \\
2 & 0 & 0 & 1.78 & 0.22 \\
3 & 0 & 0 & 0.5 & 1.5 \\
4 & 0 & 0 & 0.5 & 1.5 \\
\noalign{\smallskip}\hline
\end{tabular}
\end{table}

\section{Conclusion} \label{sec:con}
The joint use of polynomial counting functions and quadratic optimization tools makes it possible to find GMA-optimal mixed-level orthogonal arrays of a given size. 
It is worth noting that the methodology does not put any restriction on the number of levels of each factor and so it can be applied to a very wide range of designs. The methodology works with the standard partition of the set of the monomial exponents, $L = \Z_{n_1} \times \Z_{n_m}$: main effects, $2$-factor interactions, ..., $m$-factor interaction but it can also be easily adapted to work with any partition of $L$. The range of applications is limited only by the amount of computational effort required. 

\begin{acknowledgements}
The author would like to thank Paolo Brandimarte (Politecnico di Torino), Giovanni Pistone (Collegio Carlo Alberto, Moncalieri, Torino) and Maria Piera Rogantin (Universit\`{a} di Genova) for the fruitful discussions he had with them.
\end{acknowledgements}


%
%

\end{document}